\PassOptionsToPackage{unicode}{hyperref}
\PassOptionsToPackage{hyphens}{url}
\PassOptionsToPackage{dvipsnames,svgnames,x11names}{xcolor}
\documentclass[
  12pt,
  a4paper,
]{article}
\usepackage{xcolor}
\usepackage[margin=1in]{geometry}
\usepackage{amsmath,amssymb}
\usepackage{iftex}
\ifPDFTeX
  \usepackage[T1]{fontenc}
  \usepackage[utf8]{inputenc}
  \usepackage{textcomp} 
\else 
  \usepackage{unicode-math} 
  \defaultfontfeatures{Scale=MatchLowercase}
  \defaultfontfeatures[\rmfamily]{Ligatures=TeX,Scale=1}
\fi
\usepackage{lmodern}
\ifPDFTeX\else
\fi
\IfFileExists{upquote.sty}{\usepackage{upquote}}{}
\IfFileExists{microtype.sty}{
  \usepackage[]{microtype}
  \UseMicrotypeSet[protrusion]{basicmath} 
}{}
\usepackage{setspace}
\makeatletter
\@ifundefined{KOMAClassName}{
  \IfFileExists{parskip.sty}{%
    \usepackage{parskip}
  }{
    \setlength{\parindent}{0pt}
    \setlength{\parskip}{6pt plus 2pt minus 1pt}}
}{
  \KOMAoptions{parskip=half}}
\makeatother
\makeatletter
\ifx\paragraph\undefined\else
  \let\oldparagraph\paragraph
  \renewcommand{\paragraph}{
    \@ifstar
      \xxxParagraphStar
      \xxxParagraphNoStar
  }
  \newcommand{\xxxParagraphStar}[1]{\oldparagraph*{#1}\mbox{}}
  \newcommand{\xxxParagraphNoStar}[1]{\oldparagraph{#1}\mbox{}}
\fi
\ifx\subparagraph\undefined\else
  \let\oldsubparagraph\subparagraph
  \renewcommand{\subparagraph}{
    \@ifstar
      \xxxSubParagraphStar
      \xxxSubParagraphNoStar
  }
  \newcommand{\xxxSubParagraphStar}[1]{\oldsubparagraph*{#1}\mbox{}}
  \newcommand{\xxxSubParagraphNoStar}[1]{\oldsubparagraph{#1}\mbox{}}
\fi
\makeatother

\usepackage{longtable,booktabs,array}
\usepackage{calc} 
\usepackage{etoolbox}
\makeatletter
\patchcmd\longtable{\par}{\if@noskipsec\mbox{}\fi\par}{}{}
\makeatother
\IfFileExists{footnotehyper.sty}{\usepackage{footnotehyper}}{\usepackage{footnote}}
\makesavenoteenv{longtable}
\usepackage{graphicx}
\makeatletter
\newsavebox\pandoc@box
\newcommand*\pandocbounded[1]{
  \sbox\pandoc@box{#1}%
  \Gscale@div\@tempa{\textheight}{\dimexpr\ht\pandoc@box+\dp\pandoc@box\relax}%
  \Gscale@div\@tempb{\linewidth}{\wd\pandoc@box}%
  \ifdim\@tempb\p@<\@tempa\p@\let\@tempa\@tempb\fi
  \ifdim\@tempa\p@<\p@\scalebox{\@tempa}{\usebox\pandoc@box}%
  \else\usebox{\pandoc@box}%
  \fi%
}
\def\fps@figure{htbp}
\makeatother

\providecommand{\tightlist}{%
  \setlength{\itemsep}{0pt}\setlength{\parskip}{0pt}}

\usepackage[]{natbib}
\usepackage[utf8]{inputenc}
\usepackage{mathtools}
\usepackage{amsthm}
\usepackage{algorithm}
\usepackage{algpseudocode} 
\usepackage{tikz}
\usepackage{comment}
\usepackage{authblk}
\usepackage{bm}

\newcommand{\iid}{\stackrel{\text{i.i.d.}}{\sim}}

\makeatletter
\@ifpackageloaded{caption}{}{\usepackage{caption}}
\AtBeginDocument{%
\ifdefined\contentsname
  \renewcommand*\contentsname{Table of contents}
\else
  \newcommand\contentsname{Table of contents}
\fi
\ifdefined\listfigurename
  \renewcommand*\listfigurename{List of Figures}
\else
  \newcommand\listfigurename{List of Figures}
\fi
\ifdefined\listtablename
  \renewcommand*\listtablename{List of Tables}
\else
  \newcommand\listtablename{List of Tables}
\fi
\ifdefined\figurename
  \renewcommand*\figurename{Figure}
\else
  \newcommand\figurename{Figure}
\fi
\ifdefined\tablename
  \renewcommand*\tablename{Table}
\else
  \newcommand\tablename{Table}
\fi
}
\@ifpackageloaded{float}{}{\usepackage{float}}
\floatstyle{ruled}
\@ifundefined{c@chapter}{\newfloat{codelisting}{h}{lop}}{\newfloat{codelisting}{h}{lop}[chapter]}
\floatname{codelisting}{Listing}

\usepackage{amsthm}
\theoremstyle{plain}
\newtheorem{theorem}{Theorem}[section]
\theoremstyle{plain}
\newtheorem{proposition}{Proposition}[section]
\theoremstyle{definition}
\newtheorem{definition}{Definition}[section]
\theoremstyle{remark}
\AtBeginDocument{}

\newtheorem{refremark}{Remark}[section]

\makeatother
\makeatletter
\@ifpackageloaded{caption}{}{\usepackage{caption}}
\@ifpackageloaded{subcaption}{}{\usepackage{subcaption}}
\makeatother
\usepackage{bookmark}
\IfFileExists{xurl.sty}{\usepackage{xurl}}{} 
\hypersetup{
  pdftitle={Weighted Asymptotically Optimal Sequential Testing},
  pdfauthor={Bose, Bartroff},
  colorlinks=true,
  linkcolor={blue},
  filecolor={Maroon},
  citecolor={Blue},
  urlcolor={Blue},
  pdfcreator={LaTeX}}

\title{Weighted Asymptotically Optimal Sequential Testing}

\author{Soumyabrata Bose\thanks{\texttt{sbose@utexas.edu}} }
\author{Jay Bartroff\thanks{\texttt{bartroff@austin.utexas.edu}}}
\affil{
    Department of Statistics and Data Sciences \\
    University of Texas at Austin
}

\date{}
\begin{document}
\maketitle
\begin{abstract}
This paper develops a framework for incorporating prior information into
sequential multiple testing procedures while maintaining asymptotic
optimality. We define a weighted log-likelihood ratio (WLLR) as an
additive modification of the standard LLR and use it to construct two
new sequential tests: the Weighted Gap and Weighted Gap-Intersection
procedures. We prove that both procedures provide strong control of the
family-wise error rate. Our main theoretical contribution is to show
that these weighted procedures are asymptotically optimal; their
expected stopping times achieve the theoretical lower bound as the error
probabilities vanish. This first-order optimality is shown to be robust,
holding in high-dimensional regimes where the number of null hypotheses 
grows and in settings with random weights, provided that mild, 
interpretable conditions on the weight distribution are met.
\end{abstract}

\setstretch{2}
\section{Introduction}\label{introduction}

The sequential acquisition of data is a defining feature of modern
large-scale inference. In domains ranging from clinical trials \citep{Bartroff13} to
genomics \citep{Lin22} and online A/B testing \citep{Kohavi20}, data are not collected as a single
batch but accrue over time. This paradigm has fueled the development of
sequential multiple testing procedures, which aim to satisfy the dual
objectives of controlling a statistical error rate, such as the
family-wise error rate (FWER), while minimizing the expected sample size
or stopping time. A key theoretical
benchmark in this area is the concept of asymptotic optimality, where a
procedure's expected stopping time attains a fundamental
information-theoretic lower bound in the limit of vanishing error
probabilities \citep{song2017}.

While this sequential optimality framework provides a powerful theoretical
foundation, it has largely been developed under the assumption that, \textit{a-priori}, there is no differential importance, priority, or weighting of the hypotheses; in other words, that they are interchangeable. This assumption is often
misaligned with practice in applications, where substantial scientific structure or external information
is available to prioritize certain hypotheses over others. As \citet[][p.~509]{Genovese06} write, \textit{``All null hypotheses are not created equal.''}

Examples of scientific situations where naturally-occurring structure, priority,  or prior information motivates weights among the hypotheses include the following.
\begin{itemize}
\item
  In genomics, researchers conducting RNA-seq experiments to identify
  differentially expressed genes often have access to prior information
  from genome-wide association studies (GWAS) or pathway databases like
  KEGG. This information can be used to assign higher weights to genes
  with known links to the phenotype of interest, thereby prioritizing
  biologically plausible candidates \citep{korthauer19}.
\item
  In platform clinical trials, which evaluate multiple treatments
  simultaneously, prior evidence from pre-clinical studies or Phase I/II
  trials may suggest that certain drug candidates are more likely to be
  effective. Weighting these hypotheses can accelerate decisions for the
  most promising treatments, an approach with clear ethical and economic
  benefits \citep{berry15}.
\item
  In neuroscience, studies using fMRI to map brain activity may
  prioritize specific regions of interest (ROIs) based on known anatomical brain structure from pilot data or previous studies, assigning higher weights to voxels within
  those ROIs \citep{Ignatiadis16}.
\end{itemize}

In these examples, weights are naturally motivated as ``\textit{positive constants indicating the importance of the hypothesis,}'' as described by \citet{Holm79} who incorporated weights into his multiple testing procedure. Larger values of Holm's weights increase the chance of rejecting null hypotheses in his procedure. We adopt this as our working definition of weights, and consider both deterministic and random regimes for weights below in Sections~\ref{sec-methods}-\ref{sec-random-w}.

The examples above all  have explicit weights which the researchers would like to incorporate into a testing procedure.  In other settings researchers may not have explicit weights in mind but implicitly end up using a weighted procedure. For example, commonly-used Bonferroni-type multiple testing procedures compare each $p$-value $P_1,\ldots, P_J$ with a threshold and the $j$th null hypothesis is rejected if $P_j\le \alpha_j$ where $\sum_{j=1}^J \alpha_j=\alpha$, the desired overall familywise error rate. Taking weight $W_j=\alpha_j/(\alpha/J)$, this comparison is of course equivalent to comparing each $|\log P_j| +\log W_j$ with the same threshold $|\log(\alpha/J)|$.  \citet{Chen14} give an example of this selective ``$\alpha$ allocation'' in biomarker clinical trials where more biologically plausible hypotheses are assigned larger $\alpha_j$.  Another example of using implicit weights is in Bayesian testing where, for example in testing simple hypotheses $X\sim f_0$ vs.\ $X\sim f_1$, it is well known that the posterior probability~$P(X\sim f_1|X)$ is an increasing function of $\log (f_1(X)/f_0(X)) +\log W$, where $W$ is the prior odds of $f_1$. This weighted log-likelihood ratio is the statistic we utilize in our approach below (see \eqref{W.LLR}). Although primarily frequentist, in this sense our approach can be considered a way to incorporate Bayesian prior information without changing the error measure.

Regardless of how weights are motivated, in the fixed-sample setting, methods for incorporating such information
through hypothesis-specific weights are well-established and have been
shown to increase statistical power \citep{Benjamini97, Genovese06}.
However, the implications of weighting for asymptotically optimal
sequential procedures are unknown. It is unclear whether the benefits of
weighting can be achieved without compromising the guarantee of
asymptotic efficiency.

This paper addresses this gap by developing a formal framework for
weighted asymptotically optimal sequential testing. We introduce a
modified statistic, the weighted log-likelihood ratio (WLLR), to
integrate prior information directly into the evidence accumulation
process. Building on this, we propose two procedures: the Weighted Gap
Procedure for a known signal count and the Weighted Gap-Intersection
Procedure for a bounded signal count. We prove that both methods provide
strong FWE control. The central contribution of this work is to show
that they are asymptotically optimal: their expected stopping times
converge to the theoretical lower bound, demonstrating that weights
affect only lower-order terms. We further analyze the high-dimensional
regime where the number of hypotheses \(J\) grows, and a setting with
random weights, characterizing the conditions under which optimality is
preserved. Our results provide a principled methodology for leveraging
prior information to enhance the efficiency of sequential experiments.

The remainder of this paper is organized as follows.
Section~\ref{sec-setup-notations} formalizes the problem setup and
notation. Section~\ref{sec-background} provides a brief overview of
related work. In Section~\ref{sec-methods}, we introduce the weighted
sequential procedures and establish their FWE control.
Section~\ref{sec-ao} contains the core results on asymptotic optimality.
Section~\ref{sec-random-w} extends the analysis to the random weights
regime, and Section~\ref{sec-simulation} presents a simulation study to
evaluate the finite-sample performance of our methods.

\section{Setup and Notations}\label{sec-setup-notations}

We consider the problem of simultaneously testing \(J\) simple null
hypotheses against \(J\) simple alternatives. For each hypothesis
\(j \in [J] := \{1, \dots, J\}\), we observe an independent stream of
data, \(X^j = \{X_i^j, i=1, 2, \dots\}\), generated sequentially in
time. The probability measure governing the data stream \(j\), denoted
\(P^j\), is assumed to be one of two distinct possibilities: \(P_0^j\)
under the null or \(P_1^j\) under the alternative. The component testing
problem for each stream is thus \[
H_0^j: P^j = P_0^j \quad \text{vs.} \quad H_1^j: P^j = P_1^j.
\] The complete state of nature is described by the set of true
alternatives (signals), denoted by \(A \subseteq [J]\). The
complementary set, \(A^c = [J] \setminus A\), represents the set of true
nulls (noise). Throughout this work, we operate under the assumption of
prior information regarding the signal configuration, which may take the
form of either a known signal count, \(|A| = m\), or a known bound on
the signal count, \(|A| \in [l, u]\), the latter including the non-informative case $[l,u]=[0,J]$. Under the assumption of
independence across streams, the global probability measure
corresponding to a signal set \(A\) is the product measure
\(\mathbb{P}_A\) \[
\mathbb P_A = \prod_{j \in A} P_1^j \times \prod_{j \in A^c} P_0^j.
\] We let \(\mathbb E_A\) denote the expectation under this measure.

The fundamental statistic for this sequential problem is the
\emph{log-likelihood ratio} (LLR). For each stream \(j\) at time \(n\),
the LLR is defined with respect to the filtration
\(\mathcal{F}_n = \sigma(X_i^k: k \in [J], i \le n)\) as \[
\lambda^j(n) := \log \frac{d P_{1,n}^j}{d P_{0,n}^j},
\] where \(P_{i,n}^j\) is the restriction of the measure \(P_i^j\) to
\(\mathcal{F}_n\). We assume that the LLR process for each stream,
\(\{\lambda^j(n)\}_{n \ge 1}\), constitutes a random walk with
increments of finite variance. The drift of this walk is determined by
the true state of the hypothesis and is governed by the Kullback-Leibler
(KL) information, \(I_1^j := \mathbb E_1^j[\lambda^j(1)]\) and
\(I_0^j := -\mathbb E_0^j[\lambda^j(1)]\). Specifically, \[
\mathbb E_A[\lambda^j(1)] = 
\begin{cases} 
  I_1^j & \text{if } j \in A \\ 
  -I_0^j & \text{if } j \in A^c.
\end{cases}
\] We impose the standard condition that \(I_1^j > 0\) and \(I_0^j > 0\)
for all \(j \in [J]\), which ensures that evidence for the true
hypothesis accumulates over time and guarantees the eventual termination
of any well-posed sequential procedure.

The asymptotic performance of an optimal procedure is ultimately
dictated by the least informative streams within the signal and null
sets. Consequently, we define the worst-case rates of information
accumulation for a given signal configuration \(A\) as \[
\eta_1^A := \min_{k \in A} I_1^k \quad \text{and} \quad \eta_0^A := \min_{j \in A^c} I_0^j.
\]

A sequential multiple testing procedure is a pair \((T, D)\), where
\(T\) is a stopping time with respect to the filtration
\(\{\mathcal{F}_n\}_{n \ge 1}\), and \(D\) is an
\(\mathcal{F}_T\)-measurable decision rule specifying the set of
rejected hypotheses. We are interested in procedures that provide strong
control of the family-wise error rate (FWE) at pre-specified levels: Under true signal set~$A$, a procedure with decision rule~$D$ has Type~I and II FWE given by $\mathbb{P}_A(D \setminus A \neq \emptyset)$ and $\mathbb{P}_A(A \setminus D \neq \emptyset)$, respectively.  This leads to the following formal definition of the classes of
admissible procedures for our two settings.

\begin{definition}[Class of Admissible
Procedures]\protect\hypertarget{def-admissible-class}{}\label{def-admissible-class}

Let \(\alpha, \beta \in (0, 1)\) be the tolerable error probabilities.

\begin{enumerate}
\def\labelenumi{\arabic{enumi}.}
\tightlist
\item
  For a known signal count \(m\), the class of admissible procedures,
  denoted \(\Delta_{\alpha, m}\), consists of all sequential tests
  \((T, D)\) that always reject exactly \(m\) hypotheses and satisfy the
  FWE constraint
\begin{equation}\label{Delta.m.def}
\Delta_{\alpha,m} := \left\{ (T,D) : |D| = m \text{ and } \sup_{|A| = m} \mathbb{P}_A(D \neq A) \le \alpha \right\}.
\end{equation}
\item
  For a bounded signal count \(|A| \in [l, u]\), where the number of
  rejections is not fixed, the class of admissible procedures, denoted
  \(\Delta_{\alpha, \beta, l, u}\), consists of all sequential tests
  \((T,D)\) that satisfy the FWE constraints for both Type I and Type II
  errors \[
  \Delta_{\alpha, \beta, l, u} := \left\{ (T,D) : \sup_{A \in [l, u]} \mathbb{P}_A(D \setminus A \neq \emptyset) \le \alpha \text{ and } \sup_{A \in [l, u]} \mathbb{P}_A(A \setminus D \neq \emptyset) \le \beta \right\}.
  \]
\end{enumerate}

\end{definition}
Note that in \eqref{Delta.m.def} where the decision rule~$D$ is restricted to $|D|=m=|A|$,  both the Type~I and II FWE equal  $\mathbb{P}_A(D \neq A)$. The primary goal of this paper is to introduce procedures that belong to
these classes, \(\Delta_{\alpha, m}\) and
\(\Delta_{\alpha, \beta, l, u}\), and to prove that they are
asymptotically optimal. A procedure is deemed asymptotically optimal if
its expected stopping time (ESS) attains the theoretical lower bound for
its respective class in the limit as the error probabilities
\(\alpha, \beta \to 0\).

\section{Background and Related Work}\label{sec-background}

Classical multiple testing procedures focus on controlling the
family-wise error rate (FWE), the probability of making one or more
false rejections \citep{Holm79}. For the large-scale testing problems
common in modern applications, the false discovery rate (FDR) of
\citet{Benjamini95} offers a widely adopted and more powerful
alternative by controlling the expected proportion of false discoveries. See our discussion of the extension of these results to FDR in Section~\ref{sec-conclusion}.

While most research in this area assumes a fixed sample size, recent
work has addressed sequential settings where data arrive over time. In
this framework, the objective is two-fold: to control a desired error
rate and to minimize the expected sample size (ESS) required to reach a
decision. Extant sequential multiple testing procedures come in essentially two forms: Procedures that terminate individual data streams at different times \citep{Bartroff18,Bartroff14,Bartroff20,Malloy14}, and procedures that terminate sampling of all streams simultaneously to make a single, collective inference \citep{De12b,song2017,song2019,He2020}.

A key development in the latter sequential domain is the characterization of
asymptotic optimality, where a procedure's ESS is shown to match a
theoretical lower bound as the specified error rates vanish.
\citet{song2017} introduced the \emph{gap} and the
\emph{gap-intersection} procedures, which base decisions on the
separation between ordered log-likelihood ratios, and proved their
asymptotic optimality for FWE control. This optimality framework has
since been extended to other error metrics, including the FDR and its
variants \citep{song2019, He2020}.

Independent of sequential methods, the use of prior information to
improve statistical power is a well-established technique in multiple
testing. In the fixed-sample setting, this is often accomplished by
assigning weights to hypotheses. These weights can be applied directly
to modify testing procedures \citep{Holm79,Hochberg94,Genovese06}, used to define
the error metric itself, as in the weighted FDR \citep{Benjamini97}, or
integrated into a Bayesian decision-theoretic framework
\citep{Bogdan2011}.

However, the existing literature on asymptotically optimal sequential
testing has largely focused on the unweighted case, where all hypotheses
are treated as essentially exchangeable. The question of how to integrate weights
into these optimal sequential procedures, while preserving both error
control and asymptotic efficiency, remains open. This paper addresses
this question by introducing a weighted version of the gap procedure and
proving its asymptotic optimality.

\section{Weighted Sequential Procedures}\label{sec-methods}

In this section, we introduce the primary contribution of this work: a
sequential multiple testing procedure that formally incorporates prior
information through weights. We begin by addressing the foundational
case where the number of signals is known in advance, for which we
propose the Weighted Gap Procedure. In Section~\ref{weighted-gap-intersection-procedure-for-bounded-signal-count} we generalize this to case of bounds on the number of signals.

\subsection{Weighted Gap Procedure for Known Signal
Count}\label{weighted-gap-procedure-for-known-signal-count}

We build upon the unweighted gap procedure of \citet{song2017},
extending it to leverage the pre-specified weights
\(\bm{W} = (W_1, W_2, \dots, W_J)\). In this section we only assume that the $W_i$ are positive values, however see Remark~\ref{rem-scaling-wgap} about scaling. The core of our method is the
modification of the LLRs to create a set of statistics that reflect this
prior information. We define the \emph{weighted log-likelihood ratio}
(WLLR) for each stream \(j\) as an additive shift of the original LLR 
\begin{equation}\label{W.LLR}
\lambda_W^j(n) := \lambda^j(n) + \log(W_j).
\end{equation}
This formulation gives a \emph{head start} to hypotheses with larger
weights and penalizes those with smaller weights. At any time \(n\), we
denote the ordered WLLR statistics by
\(\lambda_W^{(1)}(n) \ge \lambda_W^{(1)}(n) \ge \dots \ge \lambda_W^{(J)}(n)\), and by convention set 
\begin{equation}\label{bdry.LLR.ord}
\lambda_W^{(0)}(n)=\infty\quad\mbox{and}\quad \lambda_W^{(J)}(n):=-\infty \quad\mbox{for all $n$.}
\end{equation} This leads to the following natural generalization of the unweighted gap procedure.

\begin{definition}[Weighted Gap
Procedure]\protect\hypertarget{def-wgap-procedure}{}\label{def-wgap-procedure}

For a known number of signals \(|A| = m\) and a fixed threshold
\(c > 0\), the \emph{weighted gap procedure} is the sequential test
\((T_W, D_W)\) defined by the following two components.

\begin{enumerate}
\def\labelenumi{\arabic{enumi}.}
\tightlist
\item
  \textbf{Stopping Time} \(T_W\): Sampling is stopped at the first time
  \(n\) that the gap between the \(m\)\textsuperscript{th} and
  \((m+1)\)\textsuperscript{th} ordered WLLRs is at least \(c\): \[
  T_W := \inf\left\{ n \ge 1 \mid \lambda_W^{(m)}(n) - \lambda_W^{(m+1)}(n) \ge c \right\}.
  \]
\item
  \textbf{Decision Rule} \(D_W\): Upon stopping at \(T_W\), the decision
  set consists of the indices of the hypotheses corresponding to the
  \(m\) largest WLLRs at that time.
\end{enumerate}

\end{definition}

To ensure that this procedure controls the FWE, the threshold \(c\) must
be selected appropriately. The following proposition establishes the
required condition on \(c\) by providing an upper bound on the FWE.

\begin{proposition}[FWE Control for Weighted Gap
Procedure]\protect\hypertarget{prp-fwe-control-wgap}{}\label{prp-fwe-control-wgap}

For any true signal set \(A\) with \(|A| = m\), the FWE of the weighted
gap procedure \((T_W, D_W)\) is bounded above by \[
\mathbb{P}_A(D_W \neq A) \le \exp(-c) \left(\sum_{j \in A^c} W_j\right) \left(\sum_{k \in A} \frac{1}{W_k}\right).
\] Consequently, to ensure
\(\sup_{|A| = m} \mathbb{P}_A(D_W \neq A) \le \alpha\), it suffices to
choose the threshold \(c\) such that
\begin{equation}\phantomsection\label{eq-threshold-c}{
c \ge |\log\alpha| + \log\mathcal C_W(m,J)
}\end{equation} where the term \(\mathcal C_W(m, J)\) is defined as 
\begin{equation}\label{m.CW.def}
\mathcal C_W(m,J) := \max_{|A| = m} \left\{ \left(\sum_{j \in A^c} W_j\right) \left(\sum_{k \in A} \frac{1}{W_k}\right) \right\}.
\end{equation}

\end{proposition}

The proof of Proposition~\ref{prp-fwe-control-wgap}, which relies on a
change-of-measure argument, is deferred to the Appendix. About this procedure we make
the following remarks.

\begin{refremark}[The Price of Weighting]
The term \(\mathcal C_W(m, J)\) can be interpreted as the \emph{price of
weighting}, as it quantifies the inflation of the threshold due to the
imbalance introduced by the weights. In the unweighted case where all
\(W_j = 1\), this term simplifies to \(m(J - m)\), recovering the
threshold from \citet{song2017}.

\label{rem-price-weighting}

\end{refremark}

\begin{refremark}[Computation of \(\mathcal C_W(m, J)\)]
The calculation of \(\mathcal C_W(m,J)\) appears to be a combinatorial
optimization problem over all \(\binom{J}{m}\) possible signal sets.
However, the maximum is achieved by a specific configuration. Let the
weights be ordered such that
\(W_{(1)} \le W_{(2)} \le \dots \le W_{(J)}\). The maximum is attained
by selecting the set \(A\) to be the \(m\) hypotheses corresponding to
the smallest weights. Thus, the computation simplifies to \[
\mathcal C_W(m, J) = \left(\sum_{j = m + 1}^{J} W_{(j)}\right) \left(\sum_{k = 1}^{m} \frac{1}{W_{(k)}}\right).
\] This provides a simple expression for calculating the threshold
\(c\).

\label{rem-computation-cw}

\end{refremark}

\begin{refremark}[Invariance to Scaling]
The weighted gap procedure is invariant to positive scaling of the
weights. If every weight \(W_j\) is replaced by \(\gamma W_j\) for some
constant \(\gamma > 0\), the procedure remains unchanged. This is
because both the stopping rule, which depends on the ratio
\(W_k / W_j\), and the constant \(\mathcal C_W(m,J)\) are invariant to
such scaling. This property implies that only the relative magnitudes of
the weights are meaningful, which justifies the use of a normalizing
constraint (e.g., \(\sum_{j = 1}^J W_j = J\)) to place different
weighting schemes on a common and comparable scale.

\label{rem-scaling-wgap}

\end{refremark}

\subsection{Weighted Gap-Intersection Procedure for Bounded Signal
Count}\label{weighted-gap-intersection-procedure-for-bounded-signal-count}

The Weighted Gap Procedure in Definition~\ref{def-wgap-procedure} is
designed for the specific case where the number of signals, \(m\), is
known precisely. To address the more general and often more practical
scenario where this count is only known to lie within an interval,
\(|A| \in [l, u]\), we propose a corresponding composite method. This
procedure, termed the \emph{weighted gap-intersection procedure}, adapts
the logic of its unweighted counterpart from \citet{song2017} to the
weighted setting.

\begin{definition}[Weighted Gap-Intersection
Procedure]\protect\hypertarget{def-wgi-procedure}{}\label{def-wgi-procedure}

For a bounded signal count \(|A| \in [l, u]\) and a set of positive
thresholds \((a, b, c, d)\), the \emph{weighted gap-intersection
procedure} is the sequential test \((T_{W, GI}, D_{W, GI})\) defined as
follows

\begin{enumerate}
\def\labelenumi{\arabic{enumi}.}
\tightlist
\item
  \textbf{Stopping Time} \(T_{W, GI}\): The procedure stops at the
  minimum of three individual stopping times \[
  T_{W, GI} := \min\{\tau_{1, W}, \tau_{2, W}, \tau_{3, W}\}.
  \] The component stopping times are defined as
\end{enumerate}

\begin{itemize}
\tightlist
\item
  An intersection rule, \(\tau_{2, W}\), that triggers when all WLLRs
  have exited a central indecision region, provided the number of
  positive WLLRs, \(p_W(n) := |\{j \mid \lambda_W^j(n) > 0\}|\), is
  consistent with the prior bounds: \[
  \tau_{2, W} := \inf\{n \ge 1 \mid l \le p_W(n) \le u \text{ and } \lambda_W^j(n) \notin (-a, b) \text{ for all } j \in [J]\}.
  \]
\item
  Two hybrid boundary rules \(\tau_{1, W}\) and \(\tau_{3, W}\), that
  use a gap condition to accelerate stopping when the evidence strongly
  suggests the number of signals is exactly \(l\) or \(u\): \[
  \begin{aligned}
  \tau_{1, W} &\: := \inf\{n \ge 1 \mid \lambda_W^{(l + 1)}(n) \le -a \text{ and } \lambda_W^{(l)}(n) - \lambda_W^{(l + 1)}(n) \ge c\}, \\
  \tau_{3, W} &\: := \inf\{n \ge 1 \mid \lambda_W^{(u)}(n) \ge b \text{ and } \lambda_W^{(u)}(n) - \lambda_W^{(u + 1)}(n) \ge d\}. 
  \end{aligned}
  \]
\end{itemize}

\begin{enumerate}
\def\labelenumi{\arabic{enumi}.}
\setcounter{enumi}{1}
\tightlist
\item
  \textbf{Decision Rule} \(D_{W, GI}\): Upon stopping at \(T_{W, GI}\),
  the procedure rejects the set of hypotheses corresponding to streams
  with positive WLLRs. The size of this rejection set is then truncated
  to lie within the interval \([l, u]\).
\end{enumerate}

\end{definition}

The thresholds \((a, b, c, d)\) are the parameters of the procedure that
must be selected to guarantee the desired FWE constraints. The following
proposition provides sufficient conditions for these thresholds to
ensure error control.

\begin{proposition}[FWE Control for Weighted Gap-Intersection
Procedure]\protect\hypertarget{prp-fwe-control-wgi}{}\label{prp-fwe-control-wgi}

Let the FWE be defined as the maximum probability of making at least one
false positive (Type I error) or at least one false negative (Type II
error). To ensure
\(\sup_{A \in [l, u]} \mathbb{P}_A(D_{W, GI} \setminus A \neq \emptyset) \le \alpha\)
and
\(\sup_{A \in [l, u]} \mathbb{P}_A(A \setminus D_{W, GI} \neq \emptyset) \le \beta\),
it suffices to choose the thresholds to satisfy
\begin{equation}\phantomsection\label{eq-wgi-threshold-abcd}{
\begin{aligned}
b \ge |\log(\alpha / 2)| + \log\left(\max_{A \in [l, u]} \sum_{j \in A^c} W_j \right), &\: \quad 
c \ge |\log(\alpha / 2)| + \log \mathcal C_W (l, J), \\
a \ge |\log(\beta / 2)| + \log\left(\max_{A \in [l, u]} \sum_{k \in A} W_k^{- 1} \right), &\: \quad
d \ge |\log(\beta / 2)| + \log \mathcal C_W (u, J)
\end{aligned}
}\end{equation}

\end{proposition}

We make the following remarks.

\begin{refremark}[Behavior at Interval Boundaries]
The definition of the gap-intersection procedure gracefully handles the
boundary cases where \(l = 0\) and/or \(u = J\). Recalling \eqref{bdry.LLR.ord}, if \(l = 0\) the constant
\(\mathcal C_W(0, J)\) becomes zero, causing the threshold \(c\) to be
effectively \(- \infty\). Symmetrically, if \(u = J\), then
\(\mathcal C_W(J, J) = 0\) and the threshold \(d\) becomes \(- \infty\).
This formally renders the boundary rules \(\tau_{1, W}\) and
\(\tau_{3, W}\) inactive. In the special case of no prior information on
the signal count (\(l = 0\) and \(u = J\)), the procedure thus
simplifies to the weighted intersection rule, \(T_W = \tau_{2, W}\),
with thresholds
\(a \ge |\log(\beta / 2)| + \log(\sum_{k = 1}^J W_k^{-1})\) and
\(b \ge |\log(\alpha / 2)| + \log(\sum_{j = 1}^J W_j)\). This
demonstrates that the weighted gap-intersection procedure is a proper
generalization which contains the simpler weighted intersection rule as
a special case.

\label{rem-wgi-boundary}

\end{refremark}

\begin{refremark}[Interpretation of Thresholds]
The structure of the thresholds in Proposition~\ref{prp-fwe-control-wgi}
reveals the logic of the composite stopping rule. The thresholds \(b\)
and \(a\) for the intersection rule component, \(\tau_{2, W}\), depend
on the maximum possible sum of weights (or their reciprocals) across all
admissible signal sets \(A\) such that \(|A| \in [l, u]\). In contrast,
the gap-based thresholds, \(c\) and \(d\), are more specific. They
depend on the constant \(\mathcal C_W(\cdot, J)\) evaluated only at the
respective boundaries of the interval, \(l\) and \(u\). This is because
the stopping rules \(\tau_{1, W}\) and \(\tau_{3, W}\) are designed to
handle precisely these boundary-case scenarios.

\label{rem-wgi-threshold}

\end{refremark}

\begin{refremark}[Computational Simplification]
The maximization in Proposition~\ref{prp-fwe-control-wgi} has a
closed-form solution based on the ordered weights. Let
\(W_{(1)} \le W_{(2)} \le \dots \le W_{(J)}\) denote the weights sorted
in non-decreasing order. The first maximum is attained by choosing the
set \(A\) that leaves the largest possible weights in the complement
\(A^c\), which corresponds to selecting \(A\) as the set of hypotheses
with the \(l\) smallest weights. The second maximum is attained by
choosing the set \(A\) that includes the largest possible reciprocal
weights, which corresponds to selecting \(A\) as the set of hypotheses
with the \(u\) smallest weights. This leads to the explicit formulas \[
\max_{A \in [l, u]} \sum_{j \in A^c} W_j = \sum_{j = l + 1}^{J} W_{(j)} \quad \text{and} \quad \max_{A \in [l, u]} \sum_{k \in A} W_k^{- 1} = \sum_{k = 1}^{u} W_{(k)}^{- 1}.
\] The terms \(\mathcal C_W(l, J)\) and \(\mathcal C_W(u, J)\) are
computed as described in Remark~\ref{rem-computation-cw}.

\label{rem-wgi-computation}

\end{refremark}

\section{Asymptotic Optimality Analysis}\label{sec-ao}

Having established the conditions for FWE control in the preceding
section, we now analyze the efficiency of the proposed weighted
procedures. The main result of this section is that these procedures are
asymptotically optimal. We prove this by showing that their expected
stopping times (ESS) attain the theoretical lower bounds established in
\citet{song2017} in the limit as the error probabilities vanish.

\subsection{Asymptotic Optimality of the Weighted Gap
Procedure}\label{asymptotic-optimality-of-the-weighted-gap-procedure}

We begin by proving the asymptotic optimality of the weighted gap
procedure \((T_W, D_W)\) for the known signal count setting.

\begin{theorem}[Asymptotic Optimality of the Weighted Gap
Procedure]\protect\hypertarget{thm-wgap-ao-fixedw-fixedj}{}\label{thm-wgap-ao-fixedw-fixedj}

Assume the LLR processes are random walks satisfying the assumptions in
Section~\ref{sec-setup-notations}. Let the threshold \(c\) be chosen
according to \eqref{eq-threshold-c} to guarantee FWE control at
level \(\alpha\). Then, for any true signal set
\(A \in \{A \subseteq [J]: |A| = m\}\), the weighted gap procedure is
asymptotically optimal, i.e., \[
\lim_{\alpha \to 0} \frac{\mathbb E_A (T_W)}{\frac{|\log \alpha|}{\eta_1^A + \eta_0^A}} = 1.
\]

\end{theorem}

\begin{proof}[Proof of Theorem~\ref{thm-wgap-ao-fixedw-fixedj}]
The proof proceeds by establishing an asymptotic upper bound on the
expected stopping time \(\mathbb E_A(T_W)\) and showing that it matches
the known asymptotic lower bound for any admissible procedure in this
class \citet{song2017}, which is given by \[
\inf_{(T,D) \in \Delta_{\alpha, m}} \mathbb E_A(T) \sim \frac{|\log\alpha|}{\eta_1^A + \eta_0^A} \quad \text{as } \alpha \to 0.
\] Our strategy relies on relating the stopping time \(T_W\) to a more
conservative, but more analytically tractable, stopping time. Let
\(\tilde{T}_W\) be the first time \(n\) at which the WLLR for every true
signal is greater than the WLLR for every true null by at least \(c\)
\begin{equation}\phantomsection\label{eq-wgi-conservative}{
\tilde{T}_W := \inf\left\{ n \ge 1 \mid \lambda_W^k(n) > \lambda_W^j(n) + c \quad \forall (k, j) \in A \times A^c \right\}.
}\end{equation} If the stopping condition for \(\tilde{T}_W\) is met at
time \(n\), it implies that
\(\min_{k \in A} \lambda_W^k(n) > \max_{j \in A^c} \lambda_W^j(n) + c\).
This forces the set of the top \(m\) WLLRs to be precisely the set of
true signals, \(A\). Consequently, the gap between the \(m\)-th and
\((m+1)\)-th ordered WLLRs satisfies
\(\lambda_W^{(m)}(n) - \lambda_W^{(m + 1)}(n) > c\). Therefore, the
stopping condition for \(T_W\) must also be satisfied, which implies
\(T_W \le \tilde{T}_W\). By the monotonicity of expectation, it follows
that \(\mathbb E_A(T_W) \le \mathbb E_A(\tilde{T}_W)\).

The stopping condition for \(\tilde{T}_W\) can be expressed in terms of
the unweighted LLRs. For each pair \((k, j) \in A \times A^c\), we
define the unweighted gap process
\(R_{k, j}(n) := \lambda^k(n) - \lambda^j(n)\). The stopping condition
in~\eqref{eq-wgi-conservative} is equivalent to requiring that
\(R_{k, j}(n) \ge c_{k, j}\) for all \(m (J - m)\) such pairs, where the
respective boundaries are \(\tilde c_{k, j} := c - \log(W_k / W_j)\).
Each process \(R_{k, j}(n)\) is a random walk with a positive drift
\(\mu_{k, j} = I_1^k + I_0^j\) and finite variance, as per our
assumptions.

We can now apply the result for the expected time for multiple random
walks to simultaneously cross their respective boundaries \citep[Lemma
A.2,][]{song2017}. Thus we have
\begin{align*}
\mathbb{E}_A (\tilde T_W) 
&\le \max_{(k,j) \in A \times A^c} \left( \frac{\tilde c_{k,j}}{\mu_{k,j}} \right) + O\left(m (J - m) \sqrt{\max_{(k,j) \in A \times A^c} \tilde c_{k, j}}\right) \\
&= \frac{\left|\log\alpha\right|}{\eta_1^A + \eta_0^A} + \tilde C_W + O\left(m (J - m) \sqrt{\log\alpha}\right),
\end{align*} where \[
\tilde C_W = \max_{(k,j) \in A \times A^c} \left[\frac{\log\left(\left.\max_{\left|A\right| = m} \left(\sum_{j \in A^c} W_j\right) \left(\sum_{k \in A} W_k^{- 1}\right)\right/\left(W_k / W_j\right)\right)}{I_1^k + I_0^j}\right]
\] Consequently, \[
\lim_{\alpha \to 0} \frac{\mathbb{E} (T_W)}{\frac{\left|\log\alpha \right|}{\eta_1^k + \eta_0^j}} \le \lim_{\alpha \to 0} \left[1 + \left(\tilde C_W + O\left(m (J - m) \sqrt{\left|\log\alpha \right|}\right)\right) \frac{\eta_1^k + \eta_0^j}{\left|\log\alpha\right|}\right] = 1
\] Since the upper bound on the performance of our procedure matches the
lower bound, the weighted gap procedure is asymptotically optimal. This
concludes the proof.
\end{proof}

\begin{refremark}[Finite Stopping Time]
We note that the stopping time \(T_W\) is well-defined, i.e., it is
finite almost surely. This property follows from the fact that
\(T_W \le \tilde{T}_W\), the more conservative
stopping time~\eqref{eq-wgi-conservative}. The almost sure finiteness of \(\tilde{T}_W\) is itself a
consequence of our assumption of positive KL information, which implies
that the underlying measures \(P_0^j\) and \(P_1^j\) are singular. This
ensures that each unweighted gap process \(R_{k, j}(n)\) has a positive
drift and will almost surely cross its respective boundary in finite
time.

\label{rem-wgap-finite-tw}

\end{refremark}

\begin{refremark}[Use of the Conservative Stopping Time]
The core of the preceding proof relies on bounding the procedure's
stopping time, \(T_W\), by the more conservative stopping time
\(\tilde T_W\). This approach provides an asymptotic upper bound on the
ESS that is sufficient to prove first-order optimality. However, it is
worth noting that in practice, \(T_W\) can be strictly smaller than
\(\tilde T_W\). The weighted gap procedure can stop as soon as the gap
between the \(m\)\textsuperscript{th} and
\((m + 1)\)\textsuperscript{th} ordered WLLRs is established, which does
not require every signal's WLLR to be greater than every null's. This
suggests that the procedure's finite-sample performance may be even
better than what the asymptotic upper bound alone would indicate.

\label{rem-wgap-conservative-time}

\end{refremark}

\begin{refremark}[Influence of Weight Dispersion]
The proof reveals that while the weights do not affect the first-order
asymptotic optimality, they do influence the ESS through the
second-order term \(\tilde C_W\). This term is a function of the maximal
ratio between weight-adjusted information rates. Highly disparate
weights will lead to a larger \(\tilde C_W\) and a more significant
\(O(\sqrt{|\log\alpha|})\) term. This implies that while the procedure
is always asymptotically optimal, extreme weight specifications may
degrade its finite-sample performance by increasing the expected
stopping time.

\label{rem-wgap-weight-cost}

\end{refremark}

The asymptotic optimality result
in~Theorem~\ref{thm-wgap-ao-fixedw-fixedj} holds for a fixed number of
hypotheses \(J\). A crucial question, particularly relevant for modern
large-scale applications, is whether this optimality is robust in a
high-dimensional regime where \(J\) grows as the error probability
\(\alpha\) vanishes. This type of analysis, where the problem dimensions
grow with the inverse of the error tolerance, has been explored in
related sequential contexts \citep{He2020}. We therefore investigate the
conditions on the growth rate of \(J\) under which the weighted gap
procedure remains asymptotically optimal. Let
\(\kappa := |\log\alpha|\), and consider the limit as both
\(J, \kappa \to \infty\).

\begin{theorem}[Optimality in the \(J \to \infty\)
Regime]\protect\hypertarget{thm-wgap-ao-fixedw-j}{}\label{thm-wgap-ao-fixedw-j}

The weighted gap procedure is asymptotically optimal in the regime where
\(J, \kappa \to \infty\) if the following two conditions hold:

\begin{enumerate}
\def\labelenumi{\arabic{enumi}.}
\tightlist
\item
  \(J = o(\kappa^{1/4})\),
\item
  \(\log \mathcal C_W(m, J) = o(\kappa)\).
\end{enumerate}

\end{theorem}

\begin{proof}[Proof of Theorem~\ref{thm-wgap-ao-fixedw-j}]
The upper bound on the weighted gap procedure's expected time is given
by \[
\begin{aligned}
\mathbb{E}_A (T_W) 
\le&\: \frac{\kappa + \log \mathcal C_W (m, J) + \log \max_{(k, j)\in A \times A^c}(W_j / W_k)}{\eta_1^A + \eta_0^A} + \\
&\: O\left(m(J - m) \sqrt{\kappa + \log \mathcal C_W (m, J) + \log \max_{(k, j) \in A \times A^c}(W_j / W_k)}\right).
\end{aligned}
\] Denote this bound by \(U_A (\kappa, J)\) and let $I_W + R_W = U_A(\kappa, J)\cdot (\eta_1^A + \eta_0^A)/ |\log \alpha|$ where 
\[
\begin{aligned}
I_W =&\: 1 + \left.\left(\log \mathcal C_W (m, J) + \log \max_{(k, j) \in A \times A^c} (W_j / W_k)\right)\right/\kappa, \\ 
R_W =&\: \left.(\eta_1^A + \eta_0^A) O\left(m(J - m) \sqrt{\kappa + \log \mathcal C_W (m, J) + \log \max_{(k, j) \in A \times A^c}(W_j / W_k)}\right)\right/\kappa.
\end{aligned}
\] Now, for any set \(A \subseteq [J]\) with \(|A| = m\), we can write
\[
\begin{aligned}
\mathcal C_W (m, J) \ge \left(\sum_{j \in A^c} W_j\right)\left(\sum_{k \in A} W_k^{-1}\right) \ge \max_{(k, j) \in A \times A^c} \left(\frac{W_j}{W_k}\right).
\end{aligned}
\] This implies that under condition 2, the term
\(\log \mathcal C_W (m, J) + \log \max_{(k, j)\in A \times A^c}(W_j / W_k)\)
is \(o(\kappa)\). Therefore the remainder term
\(R_W = (\eta_1^A + \eta_0^A) O(J^2 \sqrt{\kappa}) \kappa^{-1} \to 0\)
as \(J, \kappa \to \infty\). Similarly the leading term
\(I_W = 1 + o(\kappa)\kappa^{-1} \to 1\) as \(J, \kappa \to \infty\).
Thus we have \[
\lim_{J, \kappa \to \infty} \frac{\mathbb E_A (T_W)}{\frac{\kappa}{\eta_1^A + \eta_0^A}} \le 1. 
\] This concludes the proof.
\end{proof}

\begin{refremark}[Relaxed Conditions for Ordered Weights]
The condition \(\log \mathcal{C}_W(m, J) = o(\kappa)\) can be difficult
to verify directly. However, in the common scenario where weights are
assigned in a decaying, ordered fashion,
\(W_1 \ge W_2 \ge \dots \ge W_J > 0\), the condition can be
significantly simplified. In this case, \(\mathcal{C}_W(m, J)\) is
maximized by choosing \(A\) as the set of hypotheses with the smallest
weights, \(\{J - m + 1, \dots, J\}\). The term
\(\log \mathcal{C}_W(m, J)\) can be written as \[
\log \mathcal{C}_W(m, J) = \log\left(\sum_{j=1}^{J - m} W_j\right) + \log\left(\sum_{j = J - m + 1}^{J} \frac{1}{W_j}\right).
\] The first term is bounded by \(\log(J - m) + \log W_1\), which is of
order \(O(\log J)\) and thus satisfies the \(o(\kappa)\) requirement
under the condition \(J = o(\kappa^{1/4})\). Therefore, the optimality
condition effectively reduces to
\begin{equation}\phantomsection\label{eq-relaxed-cond-sum}{
\log\left(\sum_{j= J - m + 1}^{J} \frac{1}{W_j}\right) = o(\kappa).
}\end{equation} Under a further mild regularity condition on the tail of
the weight sequence, such as
\(\lim_{J \to \infty} \log W_J /  \log W_{J - m + 1} = 1\) as
\(J \to \infty\), the sum of reciprocals is dominated by its largest
term, \(1 / W_J\). This regularity condition intuitively means that the
rate of decay of the weights is stable in the extreme tail and does not
accelerate precipitously over the last few hypotheses. It is satisfied
by most common weight structures, including polynomial and exponential
decay. In this case \eqref{eq-relaxed-cond-sum} simplifies to a
more interpretable condition $-\log W_J = o(\kappa)$ on the smallest weight. This shows that for ordered, decaying weights, optimality is
maintained as long as the smallest weight does not decay to zero too
rapidly relative to the error tolerance.

\label{rem-wgap-relaxed-conditions}

\end{refremark}

\begin{refremark}[The Case of Uniformly Bounded Weights]
A simple and easily verifiable sufficient condition for optimality is
when the weights are uniformly bounded away from \(0\) and \(\infty\).
If there exist constants \(w_{\min}\) and \(w_{\max}\) such that
\(0 < w_{\min} \le W_j \le w_{\max} < \infty\) for all \(j\) as
\(J \to \infty\), then the second condition
of~Theorem~\ref{thm-wgap-ao-fixedw-j} is automatically satisfied under
the first. Specifically, in this case,
\(\mathcal C_W(m, J) \le (J - m) m \cdot w_{\max} / w_{\min}\), which
implies that \(\log \mathcal C_W(m,J)\) is of order \(O(\log J)\). The
condition \(\log J = o(\kappa)\) is weaker than and is implied by the
first condition, \(J = o(\kappa^{1/4})\). Thus, for uniformly bounded
weights, the primary constraint for maintaining asymptotic optimality is
that the number of hypotheses, \(J\), must not grow faster than the
fourth root of \(|\log\alpha|\).

\label{rem-bounded-weights}

\end{refremark}

\subsection{Asymptotic Optimality of the Weighted Gap-Intersection
Procedure}\label{asymptotic-optimality-of-the-weighted-gap-intersection-procedure}

We now extend the optimality analysis to the weighted gap-intersection
procedure, \((T_{W,GI}, D_{W,GI})\), which addresses the setting where
the signal count is bounded, \(|A| \in [l, u]\). The main result of this
section is that this more general procedure is also asymptotically
optimal, attaining the corresponding theoretical lower bound for each
case of the true signal count.

\begin{theorem}[Asymptotic Optimality of the Weighted Gap-Intersection
Procedure]\protect\hypertarget{thm-wgi-ao-fixedw-fixedj}{}\label{thm-wgi-ao-fixedw-fixedj}

Under the assumptions in Section~\ref{sec-setup-notations}, let the
thresholds \((a, b, c, d)\) be chosen according to
Proposition~\ref{prp-fwe-control-wgi} to ensure the procedure is in the
class \(\Delta_{\alpha, \beta, l, u}\). Then, for any true signal set
\(A \in \left\{A \subseteq [J] : |A| \in [l, u]\right\}\), the procedure
is asymptotically optimal. That is, its ESS attains the established
lower bound from \citet{song2017} \[
\lim_{\alpha,\beta \to 0} \frac{\mathbb E_A(T_{W,GI})}{L_A(\alpha, \beta; l, u)} = 1,
\] where \(L_A(\alpha, \beta; l, u)\) is defined as
\begin{equation}\phantomsection\label{eq-lower-lu}{
L_A(\alpha, \beta; l, u) :=
\begin{cases} 
\max\left\{ \frac{|\log\beta|}{\eta_0^A}, \frac{|\log\alpha|}{\eta_1^A + \eta_0^A} \right\} & \text{if } |A| = l, \\
\max\left\{ \frac{|\log\beta|}{\eta_0^A}, \frac{|\log\alpha|}{\eta_1^A} \right\} & \text{if } |A| \in (l, u), \\
\max\left\{ \frac{|\log\alpha|}{\eta_1^A}, \frac{|\log\beta|}{\eta_0^A + \eta_1^A} \right\} & \text{if } |A| = u.
\end{cases}
}\end{equation}

\end{theorem}

\begin{proof}[Proof of Theorem~\ref{thm-wgi-ao-fixedw-fixedj}]
The proof strategy mirrors that of
Theorem~\ref{thm-wgap-ao-fixedw-fixedj}. We derive an asymptotic upper
bound on \(\mathbb E_A(T_{W,GI})\) and show that it matches the
corresponding lower bound in \eqref{eq-lower-lu} for each case of
\(|A|\). Since
\(T_{W, GI} = \min\{\tau_{1, W}, \tau_{2, W}, \tau_{3, W}\}\), its
expectation is bounded by the minimum of the expectations of the
component stopping times,
\(\mathbb E_A[T_{W,GI}] \le \min\{\mathbb E_A[\tau_{1,W}], \mathbb E_A[\tau_{2,W}], \mathbb E_A[\tau_{3,W}]\}\).
The analysis proceeds by finding an upper bound for the relevant
stopping time in each of the three cases for the true signal count
\(|A|\).

We begin with the case \(l < \left|A\right| < u\). We have \[
T_{W,GI} = \min \{\tau_{1, W}, \tau_{2, W}, \tau_{3, W}\} \le \tau_{2, W}.
\] We define the more conservative stopping time 
$$\tilde T_{W,GI} = \inf\left\{n \ge 1 \mid (\forall k \in A, \lambda_W^k (n) \ge b) \wedge (\forall j \in A^c, \lambda_W^j (n) \le - a)\right\}
$$
which corresponds to
the first time every true signal has a weighted LLR above \(b\) and
every true null has a weighted LLR below \(- a\).  At \(\tilde T_{W,GI}\) the number of positive LLRs is
\(p_W(T_W) =\left|A\right| \in (l, u)\). This implies that at
\(\tilde T_{W,GI}\), both the conditions for \(\tau_{2, W}\) are are
satisfied and \(T_{W,GI} \le \tau_{2, W} \le \tilde T_{W,GI}\) and thus $\mathbb{E}_A (T_{W,GI}) \le \mathbb{E}_A (\tilde T_{W,GI})$. The stopping time \(\tilde T_{W,GI}\) requires a set of \(J\) random
walks to simultaneously cross their respective boundaries: for
\(k \in A\), the walk \(\lambda^k (n)\), with drift \(I_1^k\), must
cross \(b - \log W_k\) and \(j \in A^c\), the walk
\(- \lambda^{j} (n)\), with drift \(I_0^j\), must cross
\(a + \log W_j\). Using the result \citet[][Lemma~A.2]{song2017} for multiple random walks, we have \[
\begin{aligned}
\mathbb{E}_A (\tilde T_{W,GI}) 
\le&\: \max\left(\max_{k \in A} \left\{\frac{b - \log W_k}{I_1^k}\right\}, \max_{j \in A^c} \left\{\frac{a + \log W_j}{I_0^j}\right\}\right) \\
&\:+ O \left(J \sqrt{\max\left(\max_{k \in A} \left\{{b - \log W_k}\right\}, \max_{j \in A^c} \left\{{a + \log W_j}\right\}\right)}\right).
\end{aligned}
\] The first term can be simplified as \[
\begin{aligned}
\max_{k \in A} \left\{\frac{b - \log W_k}{I_1^k}\right\}
\le&\: \frac{b - \min_{k \in A}\log W_k}{\eta_1^A}, \\
\max_{j \in A^c} \left\{\frac{a + \log W_j}{I_0^j}\right\} 
\le&\: \frac{a +  \max_{j \in A^c}\log W_j}{\eta_0^A}.
\end{aligned}
\] Thus the upper bound becomes \[
\mathbb{E}_A (T_{W,GI}) \le \max\left\{\frac{b - \min_{k \in A}\log W_k}{\eta_1^A}, \frac{a +  \max_{j \in A^c}\log W_j}{\eta_0^A}\right\} + O(\sqrt{a \vee b}).
\] Now consider the case \(\left|A\right| = l\). We again
construct a more conservative stopping time 
$$\tilde T_{W,GI} = \inf \left\{n \ge 1 \mid (\forall j \in A^c, \lambda^j_W (n) \le - a) \wedge (\forall (k, j) \in A \times A^c, \lambda^k_W (n) - \lambda^j_W (n) \ge c) \right\}$$
which is the
first time that every weighted LLR for a true null is below \(- a\) and
the gap between every true signal's weighted LLR and every true null's
weighted LLR is at least \(c\). 

Since \(c > 0\), the second condition,
\(\forall (k, j) \in A \times A^c, \lambda^k_W (n) - \lambda^j_W (n) \ge c\),
implies that every true signal's LLR is strictly greater than every true
null's LLR. This forces the set of the top \(l\) ordered weighted LLRs
to be precisely the set of true signals, \(A\). Thus the condition
\(\lambda_W^{(l)} - \lambda_W^{(l + 1)} \ge c\) is met. The first
condition, \(\forall j \in A^c, \lambda^j_W (n) \le - a\), also implies
\(\lambda^{(l + 1)}_W (n) = \max_{j \in A^c} \lambda_W^{j} (n) \le - a\).
Thus, at \(\tilde T_{W,GI}\), both the conditions for \(\tau_{1, W}\)
are are satisfied and \(T_{W,GI} \le \tau_{1, W} \le \tilde T_{W,GI}\) and thus $\mathbb{E}_A (T_{W,GI}) \le \mathbb{E}_A (\tilde T_{W,GI})$. This conservative stopping time \(\tilde T_{W,GI}\) requires a
collection of \((J - l) + l(J - l) = (l + 1)(J - l)\) random walks to
cross their respective boundaries: for each of \(J - l\) true nulls
\(j \in A^c\), the random walk \(- \lambda^{j} (n)\), with drift
\(- I_0^j\), must cross the boundary \(a + \log W_j\), and for each of
the \(l (J - l)\) pairs \((k, j) \in A \times A^c\), the random walk
\(\lambda^k (n) - \lambda^j (n)\), with drift \(I_1^k + I_0^j\), must
cross the boundary \(c - \log (W_k / W_j)\). Applying \citet[][Lemma A.2]{song2017}, we have \[
\begin{aligned}
\mathbb{E}_A (\tilde T_{W,GI}) 
\le&\:\max\left(\max_{j \in A^c} \left\{\frac{a + \log W_j}{I_0^j}\right\}, \max_{(k, j) \in A \times A^c} \left\{\frac{c - \log (W_k / W_j)}{I_1^k + I_0^j}\right\}\right) + \\
O&\left((l + 1)(J - l) \sqrt{\max\left(\max_{j \in A^c} \left\{{a + \log W_j}\right\}, \max_{(k, j) \in A \times A^c} \left\{c - \log (W_k / W_j)\right\}\right)}\right).
\end{aligned}
\] Simplifying the leading term, we have \[
\mathbb{E}_A \left(T_{W,GI}\right) \le \max \left\{\frac{a + \max_{j \in A^c} \log W_j}{\eta_0^A}, \frac{c - \min_{(k, j) \in A\ \times A^c} \log (W_k / W_j)}{\eta_1^A + \eta_0^A}\right\} + O(\sqrt{a \vee c}).
\] Similar arguments can be used to derive an upper bound to
\(\mathbb{E}_A (T_{W,GI})\), when \(\left|A\right| = u\). Thus we have
\[
\mathbb{E}_A (T_{W, GI}) \le 
\begin{cases}
\max \left\{\frac{a + \max_{j \in A^c} \log W_j}{\eta_0^A}, \frac{c - \min_{(k, j) \in A\ \times A^c} \log (W_k / W_j)}{\eta_1^A + \eta_0^A}\right\} + O(\sqrt{a \vee c}), & \hspace{-0.5em}|A| = l, \\
\max\left\{\frac{b - \min_{k \in A}\log W_k}{\eta_1^A}, \frac{a +  \max_{j \in A^c}\log W_j}{\eta_0^A}\right\} + O(\sqrt{a \vee b}), & \hspace{-0.5em}|A| \in (l, u), \\
\max\left\{\frac{b - \min_{k \in A}\log W_k}{\eta_1^A}, \frac{d - \min_{(k, j) \in A\ \times A^c} \log (W_k / W_j)}{\eta_1^A + \eta_0^A}\right\} + O(\sqrt{b \vee d}), & \hspace{-0.5em}|A| = u.
\end{cases}
\] Given that the thresholds \(a, b, c, d\) are chosen to achieve FWE
control, we have \[
\lim_{\alpha,\beta \to 0} \frac{\mathbb E_A(T_{W,GI})}{L_A(\alpha, \beta; l, u)} \le 1. 
\] This concludes the proof.
\end{proof}

We now extend the high-dimensional analysis to the weighted
gap-intersection procedure. The logic parallels that of the fixed-\(m\)
case: asymptotic optimality is maintained if the second-order terms in
the ESS expansion, which depend on \(J\), are of a lower order than the
dominant \(\kappa\) term. Due to the composite nature of the stopping
time \(T_{W, GI}\), this must hold for each of its three components. The
most stringent requirement on the growth of \(J\) arises from the
boundary rules \(\tau_{1, W}\) and \(\tau_{3, W}\), which involve a
larger number of simultaneous random walk comparisons. This leads to
conditions for optimality that are analogous to those for the simpler
weighted gap procedure.

\begin{theorem}[Optimality in the \(J \to \infty\)
Regime]\protect\hypertarget{thm-wgi-ao-fixedw-j}{}\label{thm-wgi-ao-fixedw-j}

The weighted gap-intersection procedure is asymptotically optimal in the
regime where \(J, \kappa \to \infty\), where
\(\kappa = |\log (\alpha \wedge \beta)|\), if the following conditions
hold.

\begin{enumerate}
\def\labelenumi{\arabic{enumi}.}
\tightlist
\item
  \(J = o(\kappa^{1 / 4})\)
\item
  \(\max_{j \in [J]} \log W_j = o(\kappa)\)
\item
  \(- \min_{j \in [J]} \log W_j = o(\kappa)\)
\end{enumerate}

\end{theorem}

\begin{proof}[Proof of Theorem~\ref{thm-wgi-ao-fixedw-j}]
First consider the case \(|A| \in (l, u)\), and focus on the terms
\begin{align*}
b - \min_{k \in A}\log W_k
&\le |\log \alpha| + \log 2 + \log J  + \log \max_{(k, j) \in A \times A^c} \left(\frac{W_j}{W_k} \right) \\ 
a + \max_{j \in A^c} \log W_j 
&\le |\log \beta| + \log 2 + \log J  + \log \max_{(k, j) \in A \times A^c} \left(\frac{W_j}{W_k} \right).
\end{align*} Under condition 2 and 3,
\(\log \max_{(k, j) \in A \times A^c} \left({W_j} / {W_k}\right) = o(\kappa)\).
Under condition 1, \(\log J = o(\kappa)\). This implies that the
remainder term \[
R_W = O\left(J \cdot \sqrt{\max \left( \max_{k \in A} \{b - \log W_k\}, \max_{j \in A^c} \{a + \log W_j\} \right)}\right) = O(J \sqrt{\kappa}).
\] The ratio of  the leading term
\[
I_W = \max\left\{\frac{b - \min_{k \in A}\log W_k}{\eta_1^A}, \frac{a +  \max_{j \in A^c}\log W_j}{\eta_0^A}\right\} 
\] 
to the lower bound \(L_A(\alpha, \beta; l, u)\) is \[
\begin{aligned}
\frac{I_W}{L_A(\alpha, \beta; l, u)} 
&\:\le \max\left\{\frac{b - \min_{k \in A}\log W_k}{|\log\alpha|}, \frac{a + \max_{j \in A^c}\log W_j}{|\log\beta|}\right\} \\ 
\le&\: 1 + \max\left\{\frac{\log 2 + \log J + \log \frac{\max_{j \in A^c}  W_j}{\min_{k \in A} W_k}}{|\log\alpha|}, \frac{\log 2 + \log J + \log \frac{\max_{j \in A^c} W_j}{\min_{k \in A} W_k}}{|\log\beta|}\right\}, 
\end{aligned}
\] which converges to \(1\) as \(J, \kappa \to \infty\). Thus the
procedure is asymptotically optimal when \(|A| \in (l, u)\).

Now consider the case \(|A| = l\). The term
\(c - \min_{(k, j) \in A \times A^c} (W_k / W_j)\) simplifies to
\begin{multline*}
c - \min_{(k, j) \in A \times A^c} (W_k / W_j)
= |\log (\beta / 2)| + \log \mathcal C_W (l, J) + \max_{(k, j) \in A \times A^c} (W_j / W_k) \\
\le |\log (\beta / 2)|+ \log (J - l) + \log l + 2 \left[\log {\max_{j \in [J]} W_j} - \log{\min_{j \in [J]} W_k} \right]. 
\end{multline*} Under conditions 1-3,
\(c - \min_{(k, j) \in A \times A^c} (W_k / W_j) = O(\kappa)\). This
implies that the remainder term \[
\begin{aligned}
R_W 
=&\: O\left((l + 1)(J - l) \sqrt{\max\left(\max_{j \in A^c} \left\{{a + \log W_j}\right\}, \max_{(k, j) \in A \times A^c} \left\{c - \log (W_k / W_j)\right\}\right)}\right), 
\end{aligned}
\] is \(O(J^2 \sqrt{\kappa})\). The ratio of the leading term \[
I_W = \max \left\{\frac{a + \max_{j \in A^c} \log W_j}{\eta_0^A}, \frac{c - \min_{(k, j) \in A\ \times A^c} \log (W_k / W_j)}{\eta_1^A + \eta_0^A}\right\} 
\]  to the lower bound \(L_A(\alpha, \beta; l, u)\)
is \[
\begin{aligned}
\frac{I_W}{L_A(\alpha, \beta; l, u)} 
&\:\le \max\left\{\frac{c - \min_{(k, j) \in A \times A^c} (W_k / W_j)}{|\log\alpha|}, \frac{a + \max_{j \in A^c}\log W_j}{|\log\beta|}\right\} \\ 
\le&\: 1 + \max\left\{\frac{\log 2 + 2\log J + 2\log \frac{\max_{j \in A^c}  W_j}{\min_{k \in A} W_k}}{|\log\alpha|}, \frac{\log 2 + \log J + \log \frac{\max_{j \in A^c} W_j}{\min_{k \in A} W_k}}{|\log\beta|}\right\}, 
\end{aligned}
\] which converges to \(1\) as \(J, \kappa \to \infty\). Thus the
procedure is asymptotically optimal when \(|A| = l\). Similar arguments
can be used to prove the asymptotic optimality under \(|A| = u\). This
concludes the proof.
\end{proof}

\begin{refremark}[Case of Uniformly Bounded Weights]
A simple and practical sufficient condition for satisfying Conditions 2
and 3 is when the weights are uniformly bounded away from \(0\) and
\(\infty\). If there exist constants \(w_{\min}\) and \(w_{\max}\) such
that \(0 < w_{\min} \le W_j \le w_{\max} < \infty\) for all \(j\) as
\(J \to \infty\), then \(\max \log W_j\) and \(-\min \log W_j\) are
bounded by constants. In this scenario, the primary constraint for
maintaining asymptotic optimality is the condition on the growth rate of
\(J\), namely \(J = o(\kappa^{1 / 4})\).

\label{rem-bounded-weights-wgi-j}

\end{refremark}

\section{Regime of Random Weights}\label{sec-random-w}

In some applications, the weights themselves may not be fixed constants
but may be derived from external data, motivating their treatment as
random variables \citep[e.g.,][]{Genovese06,Perone04}. In this regime, we assume that a realization of the
weight vector \(\bm{W}\) is drawn from a known distribution, and
this realization is fixed before the sequential experiment begins. The
most rigorous approach is to demand that the FWE is controlled
conditionally, for every possible realization of \(\bm{W}\).

This requirement necessitates that the procedure's threshold be adaptive
to the realized weights. In the context of 
Proposition~\ref{prp-fwe-control-wgap} for the Weighted Gap Procedure with $m$ signals, for a given realization
\(\bm{W}\) we must set the threshold as a function
\[
c(\bm{W}) := \kappa + \log \mathcal{C}_W(m,J),
\] where \(\kappa = |\log\alpha|\). The goal is then to show that the
\emph{unconditional} expected stopping time, averaged over both the data
and the weights, remains asymptotically optimal. This leads to new
conditions on the moments of the weight distribution. First we focus the
analysis on the weighted gap procedure.

\begin{theorem}[Optimality with Random
Weights]\protect\hypertarget{thm-wgap-ao-w-j}{}\label{thm-wgap-ao-w-j}

Let the weight vector \(\bm{W}\) be drawn from a given distribution,
and consider the weighted gap procedure \((T_W, D_W)\) with the adaptive
threshold \(c(\bm{W}) = \kappa + \log \mathcal{C}_W(m, J)\), where
\(\kappa = |\log\alpha|\). The unconditional expected stopping time~\(\mathbb E(T_W)\) of
this procedure is asymptotically optimal in the
regime where \(J, \kappa \to \infty\) if the following conditions hold:

\begin{enumerate}
\def\labelenumi{\arabic{enumi}.}
\tightlist
\item
  \(J = o(\kappa^{1/4})\),
\item
  \(\mathbb{E}_{\bm{W}} \left[\log\left(\max_{j \in [J]} W_j\right)\right] = o(\kappa)\),
\item
  \(\mathbb{E}_{\bm{W}} \left[- \log\left(\min_{k \in [J]} W_k\right)\right] = o(\kappa)\).
\end{enumerate}

\end{theorem}

\begin{proof}[Proof of Theorem~\ref{thm-wgap-ao-w-j}]
We begin with the upper bound to the expected stopping time, conditioned
on a fixed realization of weights \(\bm{W}\) and use iterated
expectation and Jensen's inequality to get an upper bound to the
unconditional expectation: \[
\begin{aligned}
\mathbb E_{\bm{W}} \left(\mathbb E_A (T_W \mid \bm{W})\right)
\le&\: \frac{\kappa + \mathbb E_{\bm{W}} \left(\log \mathcal C_W(m, J) + \log \frac{\max_{j \in A^c} W_j}{\min_{k \in A} W_k}\right)}{\eta_1^A + \eta_0^A} + \\
&\: O\left(m (J - m) \sqrt{\kappa + \mathbb E_{\bm{W}} \left(\log \mathcal C_W(m, J) + \log \frac{\max_{j \in A^c} W_j}{\min_{k \in A} W_k}\right)}\right). 
\end{aligned}
\] Thus, proving asymptotic optimality boils down to proving \[
\mathbb {E}_{\bm{W}} \left[\log \mathcal C_W (m, J) + \log \frac{\max_{j \in A^c} W_j}{\min_{k \in A} W_k}\right] = o(\kappa). 
\] The left-hand side can be simplified as 
\begin{multline*}
\mathbb {E}_{\bm{W}} \left[\log \mathcal C_W (m, J) + \log \frac{\max_{j \in A^c} W_j}{\min_{k \in A} W_k} \right] \\
\le \log(J - m) + \log m + 2 \: \mathbb{E}_{\bm{W}} \left[\log\left(\max_{j \in [J]} W_j\right) - \log \left(\min_{k \in [J]} W_k\right)\right]. 
\end{multline*} Thus the necessary condition simplifies to \[
\mathbb{E}_{\bm{W}} \left[\log\left(\max_{j \in [J]} W_j\right) - \log \left(\min_{k \in [J]} W_k\right)\right] = o(\kappa), 
\] which is satisfied under condition 2 and 3. This concludes the proof.
\end{proof}

\begin{refremark}[A Note on the Optimality Conditions]
Conditions 2 and 3 in Theorem~\ref{thm-wgap-ao-w-j} on the weight distributions provide formal
constraints on the informativeness and potential fallibility of the
prior information encoded in the weighting scheme. Condition 2,
\(\mathbb{E}_{\bm{W}}[\log(\max_{j \in [J]} W_j)] = o(\kappa)\), is a
safeguard against the erroneous up-weighting of true nulls. A large
weight on a null hypothesis effectively lowers the stopping boundary for
its LLR, increasing the risk of a premature false rejection based on
limited data. This condition ensures that the expected effect of the
most extreme up-weighting is asymptotically negligible compared to the
primary difficulty scale \(\kappa\). Symmetrically, condition~3,
\(\mathbb{E}_{\bm{W}}[- \log(\min_{j \in [J]} W_j)] = o(\kappa)\),
protects against the excessive down-weighting of true signals. A small
weight on a signal imposes a large additive penalty on its LLR,
potentially stalling its detection and inflating the expected stopping
time. This condition ensures that the expected penalty from the most
extreme down-weighting does not asymptotically dominate the procedure.

\label{rem-note-conditions}

\end{refremark}

\begin{refremark}[Examples of Admissible Weight Distributions]
Conditions 2 and 3 of Theorem~\ref{thm-wgap-ao-w-j} constrain the tail
moments of the log-weight distribution, and hold for
distributions whose extreme order statistics have expected logarithmic
growth that is sub-linear in \(\kappa\). This primarily rules out
distributions with exceptionally heavy tails in the log-domain.

A simple example satisfying these is a \emph{binary weighting scheme}, similar to that in
\citet{Genovese06}, where each weight is independently drawn from
\(\{W_L, W_H\}\) with \(0 < W_L \le W_H < \infty\). As \(J \to \infty\),
the expected maximum log-weight,
\(\mathbb{E}_{\bm{W}} [\log(\max_{j \in [J]} W_j)]\), converges
to the constant \(\log W_H\). Since any constant is \(o(\kappa)\), the
condition holds. A symmetric argument applies to the minimum.

The conditions also hold for many common continuous distributions. If
weights are drawn from a \emph{Lognormal distribution} such that
\(\log W_j \iid \mathcal{N}(\mu, \sigma^2)\), the expected maximum of
the log-weights is of order \(O(\sqrt{\log J})\). Similarly, if weights
are drawn from a \emph{Pareto distribution}, the expected log-maximum
grows as \(O(\log J)\). Both are admissible under condition 1.

\label{rem-wgap-w-dist}

\end{refremark}

These same principles extend directly to the more general weighted
gap-intersection procedure.

\begin{theorem}[Optimality of Weighted Gap-Intersection Procedure with
Random
Weights]\protect\hypertarget{thm-wgi-ao-w-j}{}\label{thm-wgi-ao-w-j}

Let the weight vector \(\bm{W}\) be drawn from a given
distribution, and consider the weighted gap-intersection procedure
\((T_{W, GI}, D_{W, GI})\) with the adaptive thresholds
\((a(\bm{W}), b(\bm{W}), c(\bm{W}), d(\bm{W}))\)
chosen according to Proposition~\ref{prp-fwe-control-wgi} for each
realization of $\bm{W}$. The unconditional expected stopping
time~\(\mathbb E(T_{W, GI})\) of this procedure is asymptotically
optimal in the regime where \(J, \kappa \to \infty\), where
\(\kappa = |\log (\alpha \wedge \beta)|\), if the following conditions
hold:

\begin{enumerate}
\def\labelenumi{\arabic{enumi}.}
\tightlist
\item
  \(J = o(\kappa^{1/4})\)
\item
  \(\mathbb E_{\bm{W}}[\log(\max_{j \in [J]} W_j)] = o(\kappa)\)
\item
  \(\mathbb E_{\bm{W}}[-\log(\min_{j \in [J]} W_j)] = o(\kappa)\)
\end{enumerate}

\end{theorem}

The proof of this theorem  follows a similar line of
reasoning to that of Theorems~\ref{thm-wgi-ao-fixedw-j} and \ref{thm-wgap-ao-w-j}, and is thus ommitted.

\section{Simulation Study}\label{sec-simulation}

This section presents a simulation study conducted to investigate the finite-sample
performance of the weighted gap procedure. The primary goals are to
quantify the efficiency gains from incorporating informative weights, to
assess the procedure's robustness to misinformative weights, and to
examine its scalability as the number~$J$ of hypotheses grows.

\subsection{Simulation Setup}\label{simulation-setup}

We consider a Gaussian mean problem where for each stream \(j\), the
hypotheses are \(H_0^j: \mu_j = 0\) versus \(H_1^j: \mu_j = \mu\), with
observations drawn from a \(\mathcal{N}(\mu_j, 1)\) distribution. The
signal strength is fixed at \(\mu = 0.15\). The target Type I and
Type II error rates are set to \(\alpha = \beta = 0.05\). The total
number~$J$ of hypotheses  is varied from 200 to 400, while the
proportion of true signals is held constant at \(m / J = 0.1\). Each
configuration is evaluated over 10,000 Monte Carlo replications.

First, for each hypothesis \(j\), a binary ``guess'' variable
\(U_j \in \{0,1\}\) is generated, where \(U_j = 1\) represents a prior
belief that hypothesis \(j\) is a signal. The quality of these guesses
is controlled by an informativeness parameter \(\eta \ge 0\) for true signal set~$A$ as follows:
\[
\begin{aligned}
p_1 &= \mathbb{P}(U_j=1 \mid j \in A) = \frac{\eta \cdot (m/J)}{1 + (\eta-1)(m/J)}, \\
p_0 &= \mathbb{P}(U_j=1 \mid j \in A^c) = \frac{m/J}{1 + (\eta-1)(m/J)}.
\end{aligned}
\] Here, \(\eta = 1\) corresponds to uninformative guessing
(\(p_1 = p_0\)), \(\eta > 1\) corresponds to informative guessing
(\(p_1 > p_0\)), and \(\eta < 1\) corresponds to misinformative guessing
(\(p_1 < p_0\)).

The $U_j$ are incorporated into weights with a strength parameter~$r$ by setting preliminary weights to a high value
\(r \ge 1\) if \(U_j = 1\) and to a low value of 1 if \(U_j = 0\). These preliminary weights
are then normalized to ensure the mean weight is 1, yielding the final
weights \[
W_j = \frac{1 + (r - 1) U_j}{1 + (r - 1) \bar{U}},
\] where \(\bar{U} = J^{-1}\sum_{j = 1}^J U_j\). The unweighted case
corresponds to \(r = 1\), which yields \(W_j = 1\) for all \(j \in [J]\).

For each simulation run, the threshold \(c\) is calculated for the
realized weight vector~$\bm{W}$ using \eqref{eq-threshold-c}-\eqref{m.CW.def}. We compare the performance
across 4 distinct scenarios, summarized in Table
\ref{tab-sim-scenarios}. R code for this simulation study is available at \href{https://github.com/soumyabratabose/wgap-script-git}{https://github.com/soumyabratabose/wgap-script-git}. 

\begin{table}[h!]
\centering
\caption{Weighting Scenarios}
\label{tab-sim-scenarios}
\begin{tabular}{@{}lcc@{}}
\toprule
\textbf{Scenario Name} & \textbf{Informativeness ($\eta$)} & \textbf{Strength ($r$)} \\ \midrule
Unweighted & 1 & 1 \\
Informative & 20 & 5 \\
Misinformative & 0.05 & 5 \\
Noisy & 1 & 5 \\ 
\bottomrule
\end{tabular}
\end{table}

\subsection{Results}\label{results}

\begin{figure}
	\centering
	\scalebox{.75}{\input{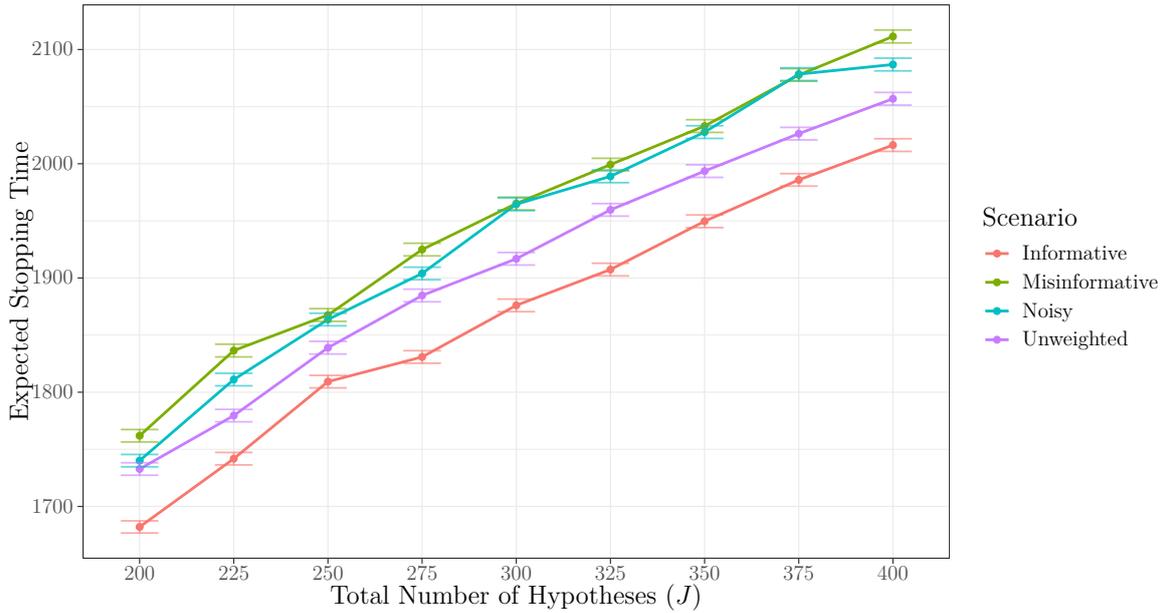}}
	\caption{Impact of Weight Informativeness on Expected Stopping Time}
	\label{plt-sim-res}
\end{figure}

The results are shown in Figure~\ref{plt-sim-res}. In the \textit{Informative} scenario, where the weights were positively correlated with the true signal status, the procedure yielded consistent and significant efficiency gains. We see that the ESS for this scenario (orange line) is uniformly lower than the \textit{Unweighted} baseline (purple line), for all tested values of \(J\). This confirms that incorporating reliable prior information directly translates into a faster experiment, as the ``head-start'' provided to true signals accelerates the decision-making process.  

The simulation also quantifies the potential costs of using poor-quality weights. In the \textit{Noisy} scenario, where the weights were uninformative but dispersed, the ESS (blue line) was consistently higher than the unweighted baseline. This increase in stopping time is an empirical demonstration of Remark~\ref{rem-wgap-weight-cost}. 

In the \textit{Misinformative} scenario, where the weights were systematically misleading, the procedure (green line) resulted in the poorest performance. This demonstrates the cost of incorporating actively incorrect prior information, which works directly against the data-driven evidence and prolongs the time to a decision.

In summary, the simulation results empirically validate our theoretical framework. The weighted gap procedure effectively leverages prior information to reduce experimental duration while suffering a predictable and quantifiable performance cost when the prior information is uninformative or misleading. 

\section{Conclusion}\label{sec-conclusion}

We have extended the theory of asymptotically optimal sequential multiple testing to the ``non-exchangeable'' setting where possible differential weights are associated with the null hypotheses. By constructing procedures based on a weighted log-likelihood ratio, we have proven that first-order asymptotic optimality is preserved. This  demonstrates that the fundamental efficiency limits of sequential testing are unaffected by the introduction of a hypothesis-specific weighting scheme.

Our analysis in the high-dimensional \(J \to \infty\) regime provides interpretable conditions on the growth rate of \(J\) and the moments of the log-weight distribution required to maintain optimality. These conditions formally constrain the permissible dispersion of prior beliefs relative to the error tolerance. In conclusion, our work provides a complete asymptotic theory for a new class of weighted sequential procedures, broadening the applicability of the optimality framework in sequential analysis.

Here we have focused on Type~I and II familywise error rates, but these results can be extended immediately to the false discovery rate (FDR) of \citet{Benjamini95}, the positive false discovery rate (pFDR), and other multiple testing error metrics using the method of \cite{He2020} since these metrics are bounded between constant multiples of FWE. For the non-asymptotic error control rates in Section~\ref{sec-methods}, the thresholds are modified by the bounding constants, and the asymptotic optimality results in Sections~\ref{sec-ao}-\ref{sec-random-w} remain unchanged.



\begin{thebibliography}{}

\bibitem[Bartroff, 2018]{Bartroff18}
Bartroff, J. (2018).
\newblock Multiple hypothesis tests controlling generalized error rates for
  sequential data.
\newblock {\em Statistica Sinica}, 28:363--398.

\bibitem[Bartroff et~al., 2013]{Bartroff13}
Bartroff, J., Lai, T.~L., and Shih, M. (2013).
\newblock {\em Sequential Experimentation in Clinical Trials: Design and
  Analysis}.
\newblock Springer, New York.

\bibitem[Bartroff and Song, 2014]{Bartroff14}
Bartroff, J. and Song, J. (2014).
\newblock {Sequential tests of multiple hypotheses controlling type I and II
  familywise error rates}.
\newblock {\em Journal of Statistical Planning and Inference}, 153:100--114.

\bibitem[Bartroff and Song, 2020]{Bartroff20}
Bartroff, J. and Song, J. (2020).
\newblock Sequential tests of multiple hypotheses controlling false discovery
  and nondiscovery rates.
\newblock {\em Sequential Analysis}, 39:65--91.

\bibitem[Benjamini and Hochberg, 1995]{Benjamini95}
Benjamini, Y. and Hochberg, Y. (1995).
\newblock {Controlling the False Discovery Rate: A Practical and Powerful
  Approach to Multiple Testing}.
\newblock {\em Journal of the Royal Statistical Society. Series B
  (Methodological)}, 57(1):289--300.

\bibitem[Benjamini and Hochberg, 1997]{Benjamini97}
Benjamini, Y. and Hochberg, Y. (1997).
\newblock {Multiple hypotheses testing with weights}.
\newblock {\em Scandinavian Journal of Statistics}, 24(3):407--418.

\bibitem[Berry et~al., 2015]{berry15}
Berry, S.~M., Connor, J.~T., and Lewis, R.~J. (2015).
\newblock The {Platform} {Trial}: {An} {Efficient} {Strategy} for {Evaluating}
  {Multiple} {Treatments}.
\newblock {\em Journal of the American Medical Association},
  313(16):1619--1620.

\bibitem[Bogdan et~al., 2011]{Bogdan2011}
Bogdan, M., Chakrabarti, A., Frommlet, F., and Ghosh, J.~K. (2011).
\newblock {Asymptotic Bayes-optimality under sparsity of some multiple testing
  procedures}.
\newblock {\em The Annals of Statistics}, 39(3):1551 -- 1579.

\bibitem[Chen et~al., 2014]{Chen14}
Chen, J.~J., Lu, T.-P., Chen, D.-T., and Wang, S.-J. (2014).
\newblock Biomarker adaptive designs in clinical trials.
\newblock {\em Translational Cancer Research}, 3(3).

\bibitem[De and Baron, 2012]{De12b}
De, S. and Baron, M. (2012).
\newblock Step-up and step-down methods for testing multiple hypotheses in
  sequential experiments.
\newblock {\em Journal of Statistical Planning and Inference}, 142:2059--2070.

\bibitem[Genovese et~al., 2006]{Genovese06}
Genovese, C.~R., Roeder, K., and Wasserman, L. (2006).
\newblock {False discovery control with $p$-value weighting}.
\newblock {\em Biometrika}, 93(3):509--524.

\bibitem[He and Bartroff, 2021]{He2020}
He, X. and Bartroff, J. (2021).
\newblock {Asymptotically optimal sequential FDR and pFDR control with (or
  without) prior information on the number of signals}.
\newblock {\em Journal of Statistical Planning and Inference}, 210:87--99.

\bibitem[Hochberg and Liberman, 1994]{Hochberg94}
Hochberg, Y. and Liberman, U. (1994).
\newblock An extended {S}imes' test.
\newblock {\em Statistics \& Probability Letters}, 21(2):101--105.

\bibitem[Holm, 1979]{Holm79}
Holm, S. (1979).
\newblock {A Simple Sequentially Rejective Multiple Test Procedure}.
\newblock {\em Scandinavian Journal of Statistics}, 6(2):65--70.

\bibitem[Ignatiadis et~al., 2016]{Ignatiadis16}
Ignatiadis, N., Klaus, B., Zaugg, J.~B., and Huber, W. (2016).
\newblock Data-driven hypothesis weighting increases detection power in
  genome-scale multiple testing.
\newblock {\em Nature Methods}, 13(7):577--580.

\bibitem[Kohavi et~al., 2020]{Kohavi20}
Kohavi, R., Tang, D., and Xu, Y. (2020).
\newblock {\em Trustworthy online controlled experiments: A practical guide to
  {A}/{B} testing}.
\newblock Cambridge University Press, Cambridge.

\bibitem[Korthauer et~al., 2019]{korthauer19}
Korthauer, K., Kimes, P.~K., Duvallet, C., Reyes, A., Subramanian, A., Teng,
  M., Shukla, C., Alm, E.~J., and Hicks, S.~C. (2019).
\newblock {A practical guide to methods controlling false discoveries in
  computational biology}.
\newblock {\em Genome Biology}, 20(1):118.

\bibitem[Lin et~al., 2022]{Lin22}
Lin, S., Scholtens, D., and Datta, S. (2022).
\newblock {\em Bioinformatics Methods: From Omics to Next Generation
  Sequencing}.
\newblock Chapman and Hall/CRC.

\bibitem[Malloy and Nowak, 2014]{Malloy14}
Malloy, M.~L. and Nowak, R.~D. (2014).
\newblock Sequential testing for sparse recovery.
\newblock {\em {IEEE} Transactions on Information Theory}, 60(12):7862--7873.

\bibitem[Perone~Pacifico et~al., 2004]{Perone04}
Perone~Pacifico, M., Genovese, C., Verdinelli, I., and Wasserman, L. (2004).
\newblock False discovery control for random fields.
\newblock {\em Journal of the American Statistical Association},
  99(468):1002--1014.

\bibitem[Song and Fellouris, 2017]{song2017}
Song, Y. and Fellouris, G. (2017).
\newblock {Asymptotically optimal, sequential, multiple testing procedures with
  prior information on the number of signals}.
\newblock {\em Electronic Journal of Statistics}, 11(1).

\bibitem[Song and Fellouris, 2019]{song2019}
Song, Y. and Fellouris, G. (2019).
\newblock {Sequential multiple testing with generalized error control: An
  asymptotic optimality theory}.
\newblock {\em The Annals of Statistics}, 47(3):1776 -- 1803.

\end{thebibliography}


\appendix

\singlespacing

\setstretch{2}

\section{Appendix}\label{sec-appendix}

\subsection{Proof of Proposition \ref{prp-fwe-control-wgap}} \label{sec-fwe-wgap-proof}

Let the error event be \(E = \left\{D_W \ne A\right\}\). This occurs if
and only if there exists at least one \(j \in A^c\) such that
\(j \in D_W\) and similarly at least one \(k \in A\) such that
\(k \notin D_W\). A necessary condition for this is
\(\lambda^j_W (T_W) \ge \lambda_W^k (T_W)\). Let \(E_{j, k}\) be the
event of such mis-ordering, then \[
E = \bigcup_{\substack{j \in A^c \\ k \in A}} E_{j, k} \quad \implies \quad \mathbb{P}_A (E) \le \sum_{j \in A^c} \sum_{k \in A} \mathbb{P}_A (E_{j, k}) .
\] Now, we will try to bound \(\mathbb{P}_A (E_{j, k})\). \(\tilde E_{j, k}\) implies \(E_{j, k}\) where \(\tilde E_{j, k}\) denotes the event
\begin{equation*}
\lambda^j_W (T_W) - \lambda_W^k(T_W) \ge c \implies \lambda^j (T_W) - \lambda^k(T_W) \ge c - \log\left(\frac{W_j}{W_k}\right) .   
\end{equation*}
Consider an alternative
hypothesis config \(C = A \cup \{j\} \backslash \{k\}\). The LLR for
testing null \(A\) versus alternative \(C\) at time \(n\) is \(
\lambda^{A, C} (n) = \lambda^j(n) - \lambda^k(n)\). Using Wald's identity, we have \[
\begin{aligned}
\mathbb{P}_A (\tilde{E}_{j, k}) 
= \mathbb{E}_C \left[\exp(- \lambda^{A, C} (T_W)) \boldsymbol{1}(\tilde E_{j, k})\right] 
\le \exp(-c) \frac{W_j}{W_k}. 
\end{aligned}
\] Thus we have \[
\mathbb{P}_A(D_W \ne A) 
\le \sum_{j \in A^c} \sum_{k \in A} \mathbb{P}_A (\tilde E_{j, k}) 
\le \exp(- c) \left(\sum_{j \in A^c} W_j\right) \left(\sum_{k \in A} \frac{1}{W_k}\right) .
\] To ensure that
\(\sup_{\left|A\right| = m} \mathbb{P}_A (D_W \ne A) \le \alpha\) ,
we set \[
\begin{aligned}
&\: \sup_{\left|A\right| = m} \left[ \exp(- c) \left(\sum_{j \in A^c} W_j\right) \left(\sum_{k \in A} \frac{1}{W_k}\right) \right] 
\le \alpha \\
\implies &\: c \ge \left|\log\alpha\right| + \max_{\left|A\right| = m} \left\{\log \left(\sum_{j \in A^c} W_j\right) \left(\sum_{k \in A} \frac{1}{W_k}\right)\right\}
\end{aligned}
\]

\subsection{Proof of Proposition \ref{prp-fwe-control-wgi}} \label{sec-fwe-wgi-proof}

The Type~I FWE is defined as
\(\operatorname{FWE}_{1, A} = \mathbb{P}_A \left(\exists j \in A^c : j \in D_{W, GI}\right)\).
This event could happen in two primary ways.

\begin{enumerate}
\def\labelenumi{\arabic{enumi}.}
\tightlist
\item
  For \(j \in A^c\), the weighted LLR becomes strongly positive
  \(\lambda_W^j \ge b\) (via \(\tau_{2, W}\) or, \(\tau_{3, W}\)).
\item
  For \(k \in A\) and \(j \in A^c\), the difference exceeds a \(c\)-gap
  \(\lambda_W^j - \lambda_W^k \ge c\) (via \(\tau_{1, W}\)).
\end{enumerate}

This implies \[
\operatorname{FWE}_{1, A} \le \sum_{j \in A^c} \mathbb{P}_A (\lambda_W^j \ge b) + \boldsymbol{1} (|A| = l)\sum_{\substack{j \in A^c \\ k \in A}} \mathbb{P}_A (\lambda_W^j - \lambda_W^k \ge c)
\]

Using change-of-measure argument, we have \[
\operatorname{FWE}_{1, A} \le \exp(-b) \left(\sum_{j \in A^c} W_j\right) + \boldsymbol 1 (\left|A\right| = l)\exp(-c) \left( \sum_{j \in A^c} W_j \right) \left( \sum_{k \in A} W_k^{-1} \right)
\]

Using similar argument for \(\operatorname{FWE}_{2, A} = \mathbb{P}_A (\exists k \in A : k \notin D_{W, GI})\) we also have \[
\begin{aligned}
\operatorname{FWE}_{2, A} 
\le\exp(- a) \left(\sum_{k \in A} W_k^{-1}\right) +
\boldsymbol 1 (\left|A\right| = u)\exp(- d) \left( \sum_{j \in A^c} W_j \right) \left( \sum_{k \in A} W_k^{-1} \right)
\end{aligned}
\]

The thresholds \(a, b, c, d\) mentioned in \eqref{eq-wgi-threshold-abcd} satisfy these two bounds.

\end{document}